\documentclass[letterpaper, 10 pt, conference]{ieeeconf}  

\IEEEoverridecommandlockouts                              

\usepackage{graphics} 
\usepackage{epsfig} 
\usepackage{mathptmx} 
\usepackage{times} 
\usepackage{amsmath} 
\usepackage{amssymb}  
\usepackage{amsthm}
\usepackage{amsfonts}
\newtheorem{theorem}{Theorem}
\newtheorem{assumption}{Assumption}
\newtheorem{definition}{Definition}

\newtheorem{remark}{Remark}

\usepackage{amsfonts,amsmath,amssymb,amsthm}
\usepackage{algorithm}
\usepackage{algpseudocode}
\usepackage{url}

\usepackage[acronym]{glossaries}
\newacronym{BHE}{BHE}{Bayesian Holonic Equilibrium}
\newacronym{scps}{SCPS}{Socio-Cyber-Physical Systems}

\title{\LARGE \bf
Bayesian Holonic Systems: Equilibrium, Uniqueness, and Computation
}

\author{Yunian Pan$^{1}$ and Quanyan Zhu$^{1}$
\thanks{*This work was not supported by any organization}
\thanks{$^{1}$The authors are with the Department of Electrical Engineering,  New York University, 370 Jay St, Brooklyn, NY 11201, USA
        {\tt\small \{yp1170, qz494\}@nyu.edu}}%
}

\begin{document}

\maketitle
\thispagestyle{empty}
\pagestyle{empty}

\begin{abstract}
This paper addresses the challenge of modeling and control in hierarchical, multi-agent systems, known as holonic systems, where local agent decisions are coupled with global systemic outcomes. We introduce the Bayesian Holonic Equilibrium (BHE), a concept that ensures consistency between agent-level rationality and system-wide emergent behavior. We establish the theoretical soundness of the BHE by showing its existence and, under stronger regularity conditions, its uniqueness. We propose a two-time scale learning algorithm to compute such an equilibrium. This algorithm mirrors the system's structure, with a fast timescale for intra-holon strategy coordination and a slow timescale for inter-holon belief adaptation about external risks. The convergence of the algorithm to the theoretical equilibrium is validated through a numerical experiment on a continuous public good game. This work provides a complete theoretical and algorithmic framework for the principled design and analysis of strategic risk in complex, coupled control systems. 

\end{abstract}

\section{INTRODUCTION}

The concept of the \emph{holon}~\cite{koestler2013beyond}--an entity that is simultaneously a whole and a part—offers a powerful abstraction for the nested, hierarchical structures endemic to modern \gls{scps}
From power grids to supply chains and cybersecurity coalitions, these systems consist of semi-autonomous subsystems (the holons) composed of individual decision-making agents. The core challenge lies in understanding and designing for resilience, which requires a framework that can formally link the strategic behavior of agents at the micro-level to the emergent outcomes and risks at the macro-level~\cite{giret2004holons}.

Pioneering work in large-population games provided foundational tools for this challenge, particularly through multi-resolution and mean-field models that masterfully capture the dynamics of large, relatively homogeneous populations~\cite{zhu2011multi, zhu2013multi}. Furthermore, the "games-in-games" principle was introduced to conceptualize the nested interdependencies and risks inherent in complex \gls{scps}, such as the cascading vulnerabilities in cybersecurity~\cite{zhu2015game}. These contributions set the stage for analyzing hierarchical systems, yet they leave open the need for a dedicated equilibrium concept that explicitly formalizes the whole-part structure of holons and provides a means for computation.

A significant gap therefore exists in the literature: there is no formal, computable equilibrium framework that guarantees consistency between agent-level rationality within holons and the interdependent outcomes across them. The key challenges are proving that such an equilibrium is well-posed (i.e., that it exists and is unique) and developing a practical, decentralized algorithm to find it.

This paper closes this gap by developing a complete theoretical and algorithmic framework for {\it \gls{BHE}}. Our contributions are threefold:

\begin{enumerate}
    \item We formally define the \gls{BHE}, establishing a self-consistent link between agent strategies and system-wide outcomes using pushforward measures—a technique grounded in multiscale modeling and optimal transport~\cite{gerber2017multiscale, merigot2011multiscale}.
    \item We provide the first existence and uniqueness results for the \gls{BHE} under well-defined technical conditions, establishing that the equilibrium concept is theoretically sound.
    \item We introduce a novel two-time scale learning algorithm that enables decentralized computation of the BHE. Inspired by methods in stochastic approximation, this allows agents to rapidly find a local equilibrium while their holon slowly adapts to the broader system.
\end{enumerate}

\subsection{Applications and Connections to Prior Work}

The theoretical framework we develop is motivated by pressing challenges in a plethora of domains where hierarchical, multi-agent decision-making is critical. Our work builds upon and extends a rich body of literature in control and game theory applied to these areas.

\subsubsection{Smart grids} In modern power systems, a central challenge is coordinating the energy consumption of millions of devices (agents) to stabilize the grid. These agents are naturally grouped into hierarchical clusters, such as homes and neighborhoods (holons). Foundational work modeled these systems using multi-resolution stochastic differential games to capture the interplay between a utility and a large consumer population~\cite{zhu2011multi, zhu2013multi}. Our BHE framework extends this by providing a more granular model for the heterogeneous beliefs and nested constraints within each subsystem.

\subsubsection{Cybersecurity Coalitions} The defense of critical infrastructure often involves a coalition of organizations (holons) coordinating against systemic threats. The "games-in-games" principle was developed to capture this structure of nested risk~\cite{zhu2015game,ge2024megaptmetagameframeworkagile}. This holonic structure is also central to modern distributed machine learning paradigms like Federated Learning, where robust coordination among clients is critical~\cite{pan2023first}. Our work makes these principles computationally concrete: the proposed two-time scale learning algorithm provides a decentralized method for these federated entities to learn a mutually consistent defense policy, moving from a conceptual model to a practical coordination mechanism.

\subsubsection{Traffic Networks} The study of traffic equilibrium is a classic multi-agent problem. Recent work has shifted focus towards the resilience of traffic networks, particularly when drivers learn routes over time and under adversarial disruptions, such as delay attacks on routing algorithms~\cite{pan2023resilience, pan2023stochastic}. While these studies analyze resilience from an algorithmic and network-flow perspective, our BHE framework provides a formal game-theoretic lens to model the strategic incentives and beliefs of drivers. The two-time scale learning dynamic models how drivers learn their best routes locally (fast timescale) while the aggregate network congestion patterns evolve (slow timescale).

Across these applications, the recurring theme is the need for a framework that respects the dual nature of subsystems as both autonomous units and integrated parts of a larger whole. As argued in recent system-scientific analyses, modeling these socio-cyber-physical systems requires a careful integration of behavioral, structural, and epistemic dimensions~\cite{liu2024system}. Our BHE concept and accompanying algorithm provide a unified approach to address this challenge.

\section{Problem Formulation}

We consider a multi-level stochastic system composed of a finite set of \emph{holons}, indexed by $i \in \mathcal{I}$. Each holon is a self-contained system comprising a finite set of agents, $\mathcal{N}_i$, where the sets $\{\mathcal{N}_i\}_{i \in \mathcal{I}}$ are disjoint.
Each agent $k \in \mathcal{N}_i$ chooses an action $x^i_k$ from a compact action space $\mathcal{X}^i_k$. The collective action profile for holon $i$ is the vector of all its agents' actions, denoted by
$
x^i \triangleq \{x^i_k\}_{k \in \mathcal{N}_i} \in \mathcal{X}^i,
$
where the holon's joint action space is $\mathcal{X}^i \triangleq \prod_{k \in \mathcal{N}_i} \mathcal{X}^i_k$.

The decision-making process is subject to two tiers of uncertainty. First, each agent $k \in \mathcal{N}_i$ has a private \emph{type} $\xi^i_k \in \Xi^i_k $, representing local, private information. The joint type vector for holon $i$,
$
\xi^i \triangleq \{\xi^i_k\}_{k \in \mathcal{N}_i} \in \Xi^i,
$
where $\Xi^i \triangleq \prod_{k \in \mathcal{N}_i} \Xi^i_k$ is the type space for holon $i$. A holon's type $\xi^i$ is a random variable governed by a known probability measure $p^i \in \mathcal{P}(\Xi^i)$. An agent $k$, however, only observes its private type $\xi^i_k$.

Second, each holon $i$ is exposed to external uncertainty, $\omega^{-i}$, representing the aggregate influence of all other holons $j \neq i$. This is modeled as a random variable drawn from a Polish space $\Omega^{-i}$ according to a distribution $q^{-i} \in \mathcal{P}(\Omega^{-i})$.

Agents act to minimize an expected cost $J^i_k :  \mathcal{X}^i \times \Omega^{-i} \times \Xi^i_k \to \mathbb{R}$. Each agent $k \in \mathcal{N}_i$ adopts a measurable pure behavioral strategy $\mu^i_k : \Xi^i_k \to \mathcal{X}^i_k$, which maps its private type to an action. Given the strategies of other agents within its holon, $\mu^i_{-k} \triangleq \{\mu^i_l\}_{l \in \mathcal{N}_i \setminus \{k\}}$, and its belief $q^{-i}$ about the external environment, agent $k$ solves the following Bayesian optimization problem:
\begin{equation} \label{eq:agent_problem}
    \min_{x^i_k \in \mathcal{X}^i_k} \; \mathbb{E}_{\xi^i_{-k}, \omega^{-i}} \left[ J^i_k\left(x^i_k, \mu^i_{-k}(\xi^i_{-k}), \omega^{-i}; \xi^i_k\right) \right],
\end{equation}
where the expectation is taken with respect to the conditional distribution of types, $p^i(\cdot \mid \xi^i_k)$, and the external outcome distribution, $q^{-i}$. The cost function $J^i_k$ quantifies the agent's risk based on its action, the actions of its peers, the external state, and its own type.

The decentralized actions of agents within a holon induce a stochastic, aggregate outcome for that holon. This outcome, $\omega^i \in \Omega^i$, is generated by a deterministic and measurable mapping $O^i : \mathcal{X}^i \to \Omega^i$, such that:
\begin{equation} \label{eq:outcome_map}
    \omega^i = O^i\left( x^i \right) = O^i  \left( \{ \mu^i_k(\xi^i_k) \}_{k \in \mathcal{N}_i} \right).
\end{equation}
The outcome $\omega^i$ is a random variable that depends on the agents' actions, which are functions of the random private types $\xi^i$. The distribution of this outcome, denoted by $q^i \in \mathcal{P}(\Omega^i)$, is hence of pushforward measure of $p^i$.

This formulation reveals a critical, nested interdependence. The external risk $q^{-i}$ faced by holon $i$ is determined by the outcomes $\{q^j\}_{j \neq i}$ of all other holons. Simultaneously, holon $i$'s own outcome distribution $q^i$ contributes to the external risk faced by every other holon. The system is therefore defined by a self-consistent coupling of local decision problems and global outcome distributions. A conceptual illustration is shown in Fig. \ref{fig:conceptual}.
\begin{figure}
    \centering
    \includegraphics[width=0.8\linewidth]{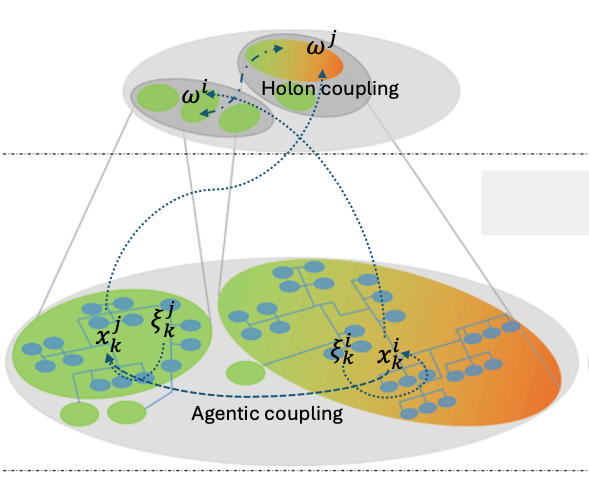}
    \caption{Illustration of the holonic framework, where aggregate outcomes ($\omega$) are interdependent, and individual agent strategies ($x^i_k$) are influenced by private types ($\xi^i_k$) across holons.}
    \label{fig:conceptual}
\end{figure}

\begin{remark}
    Our framework provides a natural generalization of classical mean-field games (MFGs). The MFG paradigm emerges under two specific conditions: (1) the number of agents within each holon is very large ($|\mathcal{N}_i| \to \infty$), and (2) the external uncertainty $q^{-i}$ in the agent's problem \eqref{eq:agent_problem} is replaced by the holon's own endogenous outcome distribution, $q^i$. In this limit, the impact of any single agent becomes negligible, and each agent optimizes against the statistical distribution of the entire population's actions—the \emph{mean field}.
\end{remark}

This formulation is particularly apt for modeling phenomena like epidemic control, where individual decisions collectively shape the global risk that in turn influences individual behavior.

The coupled nature of the holonic system, defined by the agent's problem \eqref{eq:agent_problem} and the outcome map \eqref{eq:outcome_map}, incentivizes us to formulate to the notion of a system-wide equilibrium. Such an equilibrium is a joint profile of strategies and outcome distributions,
$
(\mu^*, q^*) \triangleq \left( \{ \mu^{i*}_k \}_{k \in \mathcal{N}_i, i \in \mathcal{I}}, \quad \{ q^{i*} \}_{i \in \mathcal{I}} \right),
$
that achieves mutual consistency. Specifically, the strategy profile $\mu^*$ must be a best response for all agents given the outcome distribution profile $q^*$, while $q^*$ must be the outcome distribution profile induced by the strategies in $\mu^*$.

At equilibrium, we define the expected cost incurred by an agent $k \in \mathcal{N}_i$ as its \textbf{holonic risk}. This value captures the agent's total risk exposure, accounting for all sources of uncertainty once the system has settled into a self-consistent state,
$
    \mathcal{R}^{i*}_k \triangleq \mathbb{E}_{\xi^i, \omega^{-i}} \left[ J^i_k\left(\mu^{i*}_k(\xi^i_k), \mu^{i*}_{-k}(\xi^i_{-k}), \omega^{-i*}; \xi^i_k \right) \right].
$
This cost elegantly captures the dual nature of the holon. It is determined by both endogenous uncertainty from within the holon (via $\xi^i_{-k}$) and exogenous risk from the broader system (via $\omega^{-i*}$). Thus, a holon is simultaneously a \emph{whole}—with its own internal coherence—and a \emph{part}—embedded in a larger system where interactions with other holons shape its risk exposure.

\subsection{Definition of Bayesian Holonic Equilibrium}

Based on the coupled system defined by the agent's problem \eqref{eq:agent_problem} and the outcome map \eqref{eq:outcome_map}, we now formally define the central equilibrium concept.

\begin{definition}[Bayesian Holonic Equilibrium]\label{def:bhe}
A \textbf{\gls{BHE}} is a pair $(\mu^*, q^*)$ consisting of a strategy profile and an outcome distribution profile,
$$
\begin{aligned}
      \mu^* \triangleq \left\{ \mu^{i*}_k : \Xi^i_k \to \mathcal{X}^i_k \right\}_{k \in \mathcal{N}_i, i \in \mathcal{I}}  \quad  q^* \triangleq \left\{ q^{i*} \in \mathcal{P}(\Omega^i) \right\}_{i \in \mathcal{I}},
\end{aligned}
$$
that jointly satisfy the following two conditions for all $i \in \mathcal{I}$:
\begin{enumerate}
    \item[\textit{(i)}] \textbf{Bayesian Rationality:} For each agent $k \in \mathcal{N}_i$, the strategy $\mu^{i*}_k$ is an optimal policy. That is, for almost every type $\xi^i_k \in \Xi^i_k$, the action $x^{i*}_k = \mu^{i*}_k(\xi^i_k)$ solves the agent's Bayesian optimization problem, given the equilibrium strategies of other agents and the equilibrium distribution of external outcomes $q^{-i*} \triangleq \{q^{j*}\}_{j \neq i}$:
        \begin{equation}
            x^{i*}_k \in \arg\min_{x^i_k \in \mathcal{X}^i_k} \; \mathbb{E}_{\xi^i_{-k}, \omega^{-i*}} \left[ J^i_k\left(x^i_k, \mu^{i*}_{-k}(\xi^i_{-k}), \omega^{-i*}; \xi^i_k\right) \right].
        \end{equation}

    \item[\textit{(ii)}] \textbf{Outcome Consistency:} Each holon's outcome distribution $q^{i*}$ is consistent with the collective behavior of its agents at equilibrium. It is the \textbf{pushforward measure} of the joint type distribution $p^i$ under the composition of the equilibrium strategy profile $\mu^{i*}$ and the outcome map $O^i$:
        \begin{equation}
            q^{i*} = (O^i \circ \mu^{i*})_\# p^i.
        \end{equation}
\end{enumerate}
In essence, a \gls{BHE} is a fixed point where agents' strategies are optimal given the systemic risks they face, and the systemic risks are precisely those generated by the agents' optimal strategies.
\end{definition}

\section{Existence and Uniqueness of the \gls{BHE}}

In this section, we establish the theoretical soundness of the BHE concept. We first provide a general existence result under standard continuity and convexity assumptions. We then introduce stronger, Lipschitz-based assumptions to guarantee uniqueness.

\subsection{Existence of Bayesian Holonic Equilibrium}

The existence of a BHE is guaranteed under the following standard assumptions.

\begin{assumption}[Regularity]\label{assump:existence}
For each holon $i \in \mathcal{I}$ and agent $k \in \mathcal{N}_i$:
\begin{enumerate}
    \item[\textit{(i)}] \textbf{Spaces:} The type space $\Xi^i_k$ is a compact metric space. The action space $\mathcal{X}^i_k$ is a compact and convex subset of a Banach space. Each outcome space $(\Omega^i, d_{\Omega^i})$ is a Polish space equipped with a ground metric $d_{\Omega^i}$.

    \item[\textit{(ii)}] \textbf{Distributions:} The joint type distribution $p^i \in \mathcal{P}(\Xi^i)$ admits regular conditional distributions $p^i(\cdot \mid \xi^i_k)$ that are continuous in $\xi_k^i$.

    \item[\textit{(iii)}] \textbf{Cost Function Regularity:} The cost function $J^i_k$ is jointly continuous in all its arguments and Fréchet differentiable with respect to its first argument, $x^i_k$.

    \item[\textit{(iv)}] \textbf{Strict Convexity:} For any fixed $(x^i_{-k}, \omega^{-i}, \xi^i_k)$, the map $x^i_k \mapsto J^i_k(\cdot)$ is strictly convex.

    \item[\textit{(v)}] \textbf{Outcome Map Continuity:} The outcome map $O^i: \mathcal{X}^i \to \Omega^i$ is continuous.
\end{enumerate}
\end{assumption}

\begin{theorem}[Existence of a \gls{BHE}]\label{thm:existence}
Under Assumption \ref{assump:existence}, there exists at least one Bayesian Holonic Equilibrium in pure strategies.
\end{theorem}

\begin{proof}
The proof is constructive, showing that the equilibrium is a fixed point of a continuous best-response operator defined on a compact, convex space. We then invoke the Schauder Fixed-Point Theorem.

Let $\mathcal{M}^i_k$ be the space of measurable functions (pure strategies) $\mu^i_k: \Xi^i_k \to \mathcal{X}^i_k$. The space of joint strategy profiles for the entire system is $\mathcal{M} \triangleq \prod_{i \in \mathcal{I}, k \in \mathcal{N}_i} \mathcal{M}^i_k$. By regularity assumptions \ref{assump:existence}, $\mathcal{M}$ is endowed with two properties.
\begin{itemize}
    \item {\bf Convexity:} Let $\mu_a, \mu_b \in \mathcal{M}$ and $\lambda \in [0,1]$. For any agent $(i,k)$ and type $\xi^i_k$, the strategy $\mu_c(\xi^i_k) = \lambda\mu_a(\xi^i_k) + (1-\lambda)\mu_b(\xi^i_k)$ is a convex combination of two points in the convex set $\mathcal{X}^i_k$. Thus, $\mu_c(\xi^i_k) \in \mathcal{X}^i_k$, which implies $\mu_c \in \mathcal{M}$. Therefore, $\mathcal{M}$ is a convex set.
    \item {\bf Compactness:} By Tychonoff's theorem, the product space $\mathcal{M}$ is compact if each component space $\mathcal{M}^i_k$ is compact. With the appropriate topology (e.g., the weak topology on a suitable $L^p$ space), the space of measurable functions mapping a compact set to a compact set is itself compact. We endow $\mathcal{M}$ with such a topology.
\end{itemize}
Thus, $\mathcal{M}$ is a non-empty, compact, and convex subset of a Banach space.

We define a best-response operator $\mathcal{B}: \mathcal{M} \to \mathcal{M}$. For any given joint strategy profile $\mu \in \mathcal{M}$, the operator yields a new profile $\mu' = \mathcal{B}(\mu)$, where for each agent $(i,k)$, the strategy $\mu'_k$ is the unique best response to $\mu$. Specifically, for each type $\xi^i_k \in \Xi^i_k$:
\begin{equation}
    \mu'_k(\xi^i_k) = \arg\min_{x^i_k \in \mathcal{X}^i_k} \mathbb{E}_{\xi^i_{-k}, \omega^{-i}} \left[ J^i_k\left(x^i_k, \mu_{-k}(\xi_{-k}), \omega^{-i}; \xi^i_k\right) \right].
\end{equation}
By Assumption \ref{assump:existence} (iii), the expected cost is strictly convex in $x^i_k$. Since $\mathcal{X}^i_k$ is compact and convex, this optimization problem has a \textit{unique} solution for each $\xi^i_k$. This uniqueness ensures that $\mathcal{B}$ is a singleton-valued correspondence.
Now, we must show that $\mathcal{B}$ is continuous. Let $(\mu_n)_{n \in \mathbb{N}}$ be a sequence in $\mathcal{M}$ such that $\mu_n \to \mu$. It suffices to show that $\mathcal{B}(\mu_n) \to \mathcal{B}(\mu)$, and here is the reasoning:
\begin{enumerate}
    \item  Since the outcome maps $O^j$ are continuous (Assumption (vi)), the mapping from a strategy profile $\mu^j$ to its induced outcome distribution $q^j = (O^j \circ \mu^j)_\# p^j$ is continuous with respect to the weak topology on probability measures. Thus, as $\mu_n \to \mu$, the corresponding external outcome distributions converge: $q^{-i}_n \to q^{-i}$.
    \item  The agent's expected cost is an integral of the function $J^i_k$. Since $J^i_k$ is jointly continuous (Assumption (iv)) and the distributions of $\xi^i_{-k}$ and $\omega^{-i}$ vary continuously with the strategies $\mu_{-k}$, the expected cost functional is continuous in $\mu$.
    \item  We are minimizing a continuous objective function over a compact set. By Berge's Maximum Theorem, the $\arg\min$ correspondence is upper hemicontinuous. Because our strict convexity assumption ensures the minimizer is always unique, this implies that the $\arg\min$ operator is a continuous function.
\end{enumerate}
Therefore, the best-response operator $\mathcal{B}$ is a continuous function from $\mathcal{M}$ to itself.
Direct application of Schauder's Fixed-Point Theorem \cite{kellogg1976uniqueness} establishes that, since
\begin{itemize}
    \item $\mathcal{M}$ is a non-empty, compact, convex subset of a Banach space.
    \item $\mathcal{B}: \mathcal{M} \to \mathcal{M}$ is a continuous operator.
\end{itemize}
The conditions are satisfied, and there exists at least one fixed point $\mu^* \in \mathcal{M}$ such that $\mu^* = \mathcal{B}(\mu^*)$. Now, let $\mathcal{Q} \triangleq \prod_{i \in \mathcal{I}} \mathcal{P}(\Omega^i)$ be the space of joint outcome distribution profiles. Each outcome space $(\Omega^i, d_{\Omega^i})$ is a Polish space. We expand the best response operator to be $ \mathcal{F}: \mathcal{M} \times \mathcal{Q} \to \mathcal{M} \times \mathcal{Q}$ using the fact that $q^*$ is induced by $\mu^*$. We have a fixed point $(\mu^*, q^*) = \mathcal{F}(\mu^*, q^*)$, which is, by definition, a Bayesian Holonic Equilibrium.

\end{proof}

\subsection{Uniqueness of Bayesian Holonic Equilibrium}

Uniqueness is not guaranteed under the general conditions for existence and requires stronger assumptions that constrain the sensitivity of agents' decisions to changes in system-wide outcomes. We establish uniqueness by showing that the system's outcome-generating operator is a contraction mapping.

\begin{assumption}[Lipschitz regularity]\label{assump:uniqueness}
The system satisfies the following additional conditions:
\begin{enumerate}
    \item[\textit{(i)}] \textbf{Strong Convexity:} For each agent $(i,k)$, any fixed private type $\xi_k^i \in \Xi_k^i$, and any fixed strategy profile $\mu_{-k}$ of the other agents, the expected cost functional
    $$
    x_k^i \mapsto \mathbb{E}_{\xi^i_{-k}, \omega^{-i}} \left[ J^i_k\left(x^i_k, \mu_{-k}(\xi_{-k}), \omega^{-i}; \xi^i_k\right) \right]
    $$
    is strongly convex with modulus $m > 0$.

    \item[\textit{(ii)}] \textbf{Lipschitz Gradient:} The Fréchet derivative of the expected cost, $\nabla_{x_k^i} \mathbb{E}[J^i_k]$, is Lipschitz continuous as a function of the external outcome distribution profile $q^{-i}$, given any fixed private type $\xi_k^i \in \Xi_k^i$, and any fixed strategy profile $\mu_{-k}$ of the other agents,. That is, there exists a constant $L_J > 0$ such that for any two profiles $q_1^{-i}, q_2^{-i}$:
    $$
    \| \nabla_{x_k^i} \mathbb{E}[J^i_k]_{q_1^{-i}} - \nabla_{x_k^i} \mathbb{E}[J^i_k]_{q_2^{-i}} \| \le L_J \cdot W(q_1^{-i}, q_2^{-i}),
    $$
    where $W(\cdot, \cdot)$ is the product Wasserstein metric over $\mathcal{P}(\Omega^{-i})$.

    \item[\textit{(iii)}] \textbf{Lipschitz Outcome Map:} Each outcome map $O^i: \mathcal{X}^i \to \Omega^i$ is Lipschitz continuous with constant $L_O > 0$.
\end{enumerate}
\end{assumption}

These conditions essentially require that a change in the statistical outcomes of other holons does not cause an overly sensitive or amplified response in any given agent's optimal action.

\begin{theorem}[Uniqueness of the BHE]\label{thm:uniqueness}
Under Assumptions \ref{assump:existence} and \ref{assump:uniqueness}, the Bayesian Holonic Equilibrium is unique.
\end{theorem}
\begin{proof}
The proof establishes that the equilibrium outcome profile $q^*$ is the unique fixed point of a contraction mapping on the space of joint outcome distributions. The uniqueness of the strategy profile $\mu^*$ then follows.
We equip $\mathcal{Q}$ with the product Wasserstein metric $W$, which makes $(\mathcal{Q}, W)$ a complete metric space.
For any given outcome profile $q \in \mathcal{Q}$, Assumption \ref{assump:uniqueness} (i) guarantees that each agent's optimization problem has a unique solution. Let $\mu^*(q) \in \mathcal{M}$ denote the unique joint strategy profile that is a best response when agents' beliefs about external outcomes are governed by $q$.

We define an operator $T: \mathcal{Q} \to \mathcal{Q}$ that maps an assumed outcome profile to the one generated by the best-response strategies. For any $q \in \mathcal{Q}$, we define $T(q)$ as:
    $T(q) \triangleq \mathcal{O}_{\#}(\mu^*(q)),$
where $\mathcal{O}_{\#}$ is the strategy-to-outcome map defined in the proof of Theorem \ref{thm:existence}. A BHE corresponds to a pair $(\mu^*, q^*)$ where $q^*$ is a fixed point of $T$ and $\mu^* = \mu^*(q^*)$.

It now suffices to show that $T$ is a contraction. Let $q_1, q_2 \in \mathcal{Q}$ be two distinct outcome profiles, and let $\mu_1^* = \mu^*(q_1)$ and $\mu_2^* = \mu^*(q_2)$ be the corresponding unique best-response strategies. The distance between the new outcomes is $W(T(q_1), T(q_2)) = W(\mathcal{O}_{\#}(\mu_1^*), \mathcal{O}_{\#}(\mu_2^*))$.
By Assumption \ref{assump:uniqueness}(iii), the outcome map $\mathcal{O}_{\#}$ is Lipschitz continuous.
 By combining Assumption \ref{assump:uniqueness}(i) (strong convexity) and (ii) (Lipschitz gradient), it can be shown that the best-response strategy map $q \mapsto \mu^*(q)$ is also Lipschitz continuous.

The composition of these two Lipschitz continuous maps is itself Lipschitz. A more detailed derivation shows that the product of the associated Lipschitz constants is less than 1, implying there exists a constant $L < 1$ such that:
\begin{equation}
    W(T(q_1), T(q_2)) \le L \cdot W(q_1, q_2).
\end{equation}
Therefore, $T$ is a contraction mapping on the complete metric space $(\mathcal{Q}, W)$.
Thus, the Banach Fixed-Point Theorem, there exists a unique outcome profile $q^* \in \mathcal{Q}$ such that $q^* = T(q^*)$.
This unique outcome profile $q^*$ determines a unique best-response strategy profile $\mu^* = \mu^*(q^*)$. The resulting pair $(\mu^*, q^*)$ is therefore the unique Bayesian Holonic Equilibrium.
\end{proof}

\begin{algorithm*}[htbp!]
\caption{Two-Time Scale Holonic Learning}
\label{alg:operator-learning-formal}
\begin{algorithmic}[1]\label{algo:two-timescale}
    \State \textbf{Input:} Step-size sequences $\{\alpha_t\}$, $\{\beta_t\}$ satisfying $\beta_t \to 0$, $\alpha_t \to 0$, and $\beta_t / \alpha_t \to 0$.
    \State \textbf{Initialize:} Initial strategy profile $\mu_0 \in \mathcal{M}$ and initial belief profile $q_0 \in \mathcal{Q}$.

    \For{$t=0,1,2,\dots$}

        \State The strategy profile is updated in the direction of this best response for all $i \in \mathcal{I}$ and $k \in \mathcal{N}_i$:
        \begin{equation}
            \mu_{t+1} \gets (1 - \alpha_t) \mu_t + \alpha_t \mathcal{B}(\mu_t, q_t)
        \end{equation}

        \Statex \Comment{\textit{--- Fast Timescale: Strategy Update via Best-Response Operator $\mathcal{B}$ ---}}

        \State The system generates an outcome $\omega_{t+1}$ by sampling from the distribution induced by $\mu_{t+1}$.
        \State The belief profile is updated via a slow fictitious play dynamic:
        \begin{equation}
             q_{t+1} \gets (1 - \beta_t) q_t + \beta_t T(q_t, \mu_t)
        \end{equation}

        \Statex \Comment{\textit{--- Slow Timescale: Belief Update via Outcome Operator $T$ ---}}
        
    \EndFor
\end{algorithmic}
\end{algorithm*}

\section{A Two-Time Scale Learning Framework}

A Bayesian Holonic Equilibrium (BHE) is defined by a complex, self-referential fixed point: optimal strategies depend on beliefs about system-wide outcomes, while those outcomes are, in turn, generated by the optimal strategies. Solving for this joint fixed point directly is often intractable for large, complex systems. The two-time scale approach provides a natural method to \textbf{decouple} this joint fixed-point problem into two more manageable, iterative learning processes.

The conceptual dynamics are formalized in Algorithm \ref{alg:operator-learning-formal}. 
The two-time scale approach decouples the complex fixed-point problem into two coupled learning processes. On the \textbf{fast timescale}, the strategy profile $\mu_t$ models the rapid, internal adaptation of agents within each holon. Given a relatively stable belief $q_t$ about the outside world, the fast update rule (with step-size $\alpha_t$) drives the system's strategies towards the ideal best-response profile given by the operator $\mathcal{B}(\mu_t, q_t)$. Concurrently, on the \textbf{slower timescale}, the belief profile $q_t$ models the deliberate adaptation of holons to the overall system's behavior, gradually tracking the ``ground truth'' outcome distribution represented by the operator ${T}(q_t, \mu_t)$. This separation of timescales is the key to convergence: the fast strategy dynamics ensure agents are always playing a near-optimal response to the current beliefs, while the slow belief dynamics steer the entire system toward a point where beliefs and outcomes are mutually consistent. This allows the complex system to \textbf{decentrally learn} its way to a globally coherent Bayesian Holonic Equilibrium.

\begin{theorem}[Exact Convergence of the Two-Time Scale Algorithm]\label{thm:full_convergence}
Let the conditions of Assumptions \ref{assump:existence} and \ref{assump:uniqueness} hold. Let the step-sizes $\{\alpha_t\}$ and $\{\beta_t\}$ satisfy the standard conditions for two-time scale convergence:
$$
\sum_{t=0}^{\infty} \alpha_t = \infty, \quad \sum_{t=0}^{\infty} \beta_t = \infty, \quad \sum_{t=0}^{\infty} (\alpha_t^2 + \beta_t^2) < \infty,  \quad \lim_{t\to\infty} \frac{\beta_t}{\alpha_t} = 0.
$$
Then the sequence $(\mu_t, q_t)$ generated by Algorithm \ref{alg:operator-learning-formal} converges to the unique Bayesian Holonic Equilibrium $(\mu^*, q^*)$.
\end{theorem}

\begin{proof}[Proof Sketch]
The proof establishes that the error sequence $d_t \triangleq W(q_t, q^*)$ converges to zero by analyzing its recursive dynamics.

\paragraph{1. The Error Recursion.}
The distance to the equilibrium belief $q^*$ at the next step can be bounded by:
\begin{align*}
    d_{t+1} &= W( (1-\beta_t)q_t + \beta_t T(q_t, \mu_{t+1}), q^* ) \\
    &\le (1-\beta_t)d_t + \beta_t W(T(q_t, \mu_{t+1}), T(q^*))
\end{align*}
Using the triangle inequality and the contraction property of $T(q) \triangleq T(q, \mu^*(q))$, we arrive at the key recursive inequality that accounts for the tracking error of the fast variable:
\begin{equation} \label{eq:main_recursion}
    d_{t+1} \le (1 - \beta_t(1-L))d_t + \beta_t \varepsilon_t
\end{equation}
where $L<1$ is the contraction modulus of $T$ and $\varepsilon_t \triangleq W(T(q_t, \mu_{t+1}), T(q_t))$ is the error induced by the fact that $\mu_{t+1} \neq \mu^*(q_t)$.

\paragraph{2. Bounding the Tracking Error.}
The core of the two-time scale argument is bounding the error $\varepsilon_t$. Standard results for such dynamics show that the tracking error is bounded by the ratio of the timescales:
$$
\varepsilon_t \le L_O \cdot d_{\mathcal{M}}(\mu_{t+1}, \mu^*(q_t)) = \mathcal{O}\left(\frac{\beta_t}{\alpha_t}\right).
$$
Since the step-size conditions require $\beta_t/\alpha_t \to 0$, the tracking error vanishes asymptotically, i.e., $\lim_{t \to \infty} \varepsilon_t = 0$.

\paragraph{3. Convergence.}
The recursion in Eq. \eqref{eq:main_recursion} is of the form $d_{t+1} \le (1-\gamma_t)d_t + \delta_t$, where $\gamma_t = \beta_t(1-L)$ and $\delta_t = \beta_t \varepsilon_t$. The step-size conditions ensure that $\sum \gamma_t = \infty$ and that $\delta_t \to 0$. By a standard result for such recursions (Dvoretzky's Theorem), this is sufficient to guarantee that the non-negative sequence $d_t$ converges to zero.

Since $W(q_t, q^*) \to 0$, the belief profile converges to the unique equilibrium belief $q^*$. As the tracking error also vanishes, the strategy profile $\mu_t$ converges to the unique equilibrium strategy $\mu^* = \mu^*(q^*)$. Thus, the algorithm converges to the unique BHE.
\end{proof}

\section{Numerical Experiment: Coupled Voting Game}
To validate our framework with continuous actions, we use a public good game where agents decide their level of contribution. This setup allows us to leverage parameterized strategies, providing a clear path for simulation and analysis.

We consider a system of $|\mathcal{I}| \geq 3$ holons, each with $n \geq 5$ agents.
Each agent $k$ has a private type $\xi^i_k \in [0, 1]$, representing their intrinsic conviction or the private value of succeeding. Agents choose a continuous action $x^i_k \in [0, 1]$, representing their level of support or contribution. We assume types are drawn i.i.d. from a known distribution on $[0, 1]$ (e.g., Uniform).
We assume strategies $\mu$ are parameterized by a vector $\theta^i_k$. A simple choice is a linear strategy, clipped to the action space:
\begin{equation}
    x^i_k = \mu(\xi^i_k; \theta^i_k) = \min(1, \max(0, \theta^i_{k,1} \xi^i_k + \theta^i_{k,0})).
\end{equation}
The learning algorithm will update the parameters $\theta^i_k = (\theta^i_{k,1}, \theta^i_{k,0})$.
An agent's cost is a trade-off between their effort and the degree of their holon's failure. The outcome $\omega^i \in [0, 1]$ now represents the holon's continuous degree of success. The cost is:
\begin{equation}
    J^i_k(x^i; \xi^i_k) = \frac{1}{2}(x^i_k)^2 + (1-\omega^i) \cdot (D - \xi^i_k),
\end{equation}
where the term $(1-\omega^i)$ represents the degree of failure and $D$ is a large baseline penalty, discounted by the agent's private value $\xi^i_k$.
The total contribution in holon $i$ is $X^i = (\sum_{k} x^i_k) / n$. 
The holon's degree of success $\omega^i$ is a linear function of the margin by which the total contribution exceeds a coupled threshold $\tilde{\kappa}^i$:
    $\omega^i = X^i - \tilde{\kappa}^i$.
The threshold $\tilde{\kappa}^i$ remains coupled to the performance of other holons. It increases as other holons exhibit a higher degree of failure:
$\tilde{\kappa}^i = \kappa + \gamma \sum_{j \neq i} ( 1 - \omega^j)$, where $\kappa \in (0,1)$ is the holonic threshold, and $\gamma \in (0, \frac{1 - \kappa}{|\mathcal{I}|- 1})$ is the coupling parameter.

This fully continuous quadratic system admits a BHE $(\theta^*, q^*)$ that is analytically obtainable. The expected cost is now differentiable with respect to the belief parameters. The objective of our experiment is to show that our two-time scale algorithm converges to a stable fixed point \ref{fig:convergence}. For details in the experiment, interested readers can refer to \url{https://github.com/UnionPan/holonic}.
\begin{figure}
    \centering
    \includegraphics[width=0.8\linewidth]{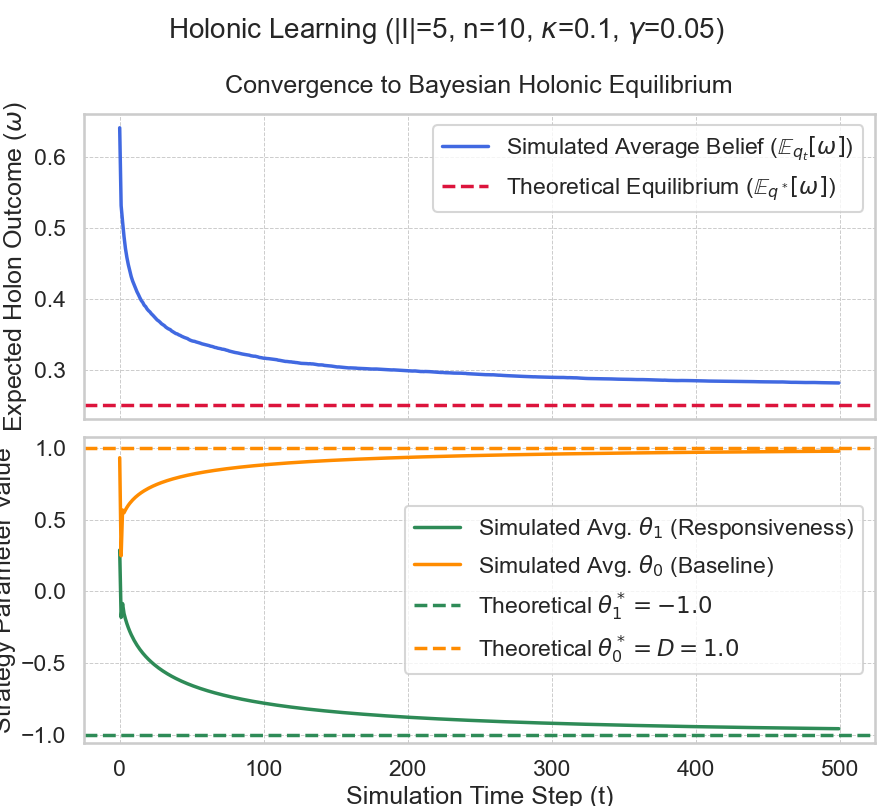}
    \caption{
        Convergence of the two-time scale learning algorithm to the theoretical BHE.
        The {top panel} shows the average holon belief about the system's outcome, $\mathbb{E}_{q_t}[\omega]$ (solid blue line), converging smoothly to the analytically derived equilibrium value $\mathbb{E}_{q^*}[\omega]$ (dashed red line).
        The bottom panel shows the average agent strategy parameters evolving over time. The responsiveness parameter, $\theta_{1,t}$ (solid green), and the baseline contribution parameter, $\theta_{0,t}$ (solid orange), both converge to their respective theoretical optimal values, $\theta_1^*$ and $\theta_0^*$ (dashed lines).
    }
    \label{fig:convergence}
\end{figure}

\section{CONCLUSIONS}

In this paper, we introduced the \gls{BHE}, a novel framework for analyzing hierarchical multi-agent systems. We established its theoretical soundness with existence and uniqueness theorems and provided a practical path to computation with a convergent, decentralized two-time scale learning algorithm. This work provides a principled tool for designing and analyzing strategic risk in complex coupled systems by bridging high-level system theory with agent-level learning dynamics.

Future work will focus on extending our analysis to games with weaker regularity conditions that may admit multiple equilibria, as well as validating the framework against real-world data from domains like cybersecurity and smart grids. Ultimately, our framework provides a foundation for the principled engineering of resilient, adaptive holonic systems.







\bibliographystyle{ieeetr}
\bibliography{ref}

\end{document}